\documentclass{article}

\usepackage[utf8]{inputenc}
\usepackage[T1]{fontenc}
\usepackage{lmodern}
\usepackage{microtype}
\usepackage{amsmath,amssymb,amsthm}
\usepackage{graphicx}
\usepackage[margin=1.2in]{geometry}
\usepackage{algorithm}
\usepackage{algpseudocode}
\usepackage[hidelinks]{hyperref}
\usepackage[capitalise,noabbrev]{cleveref}
\usepackage{color}

\newcommand{\Aall}{A_{\mathrm{all}}}
\newcommand{\Queries}{Q}

\newcommand{\Feasy}{\mathcal{F}_{c\sqrt{n/\log n}}}
\newcommand{\Feasywarm}{\mathcal{F}_{c_0\sqrt{n}/\log n}}
\newcommand{\OPT}{\mathrm{OPT}}
\newcommand{\Type}{\mathrm{Type}}
\newcommand{\type}{\mathrm{type}}
\DeclareMathOperator*{\Ex}{\mathbb{E}}

\newtheorem{theorem}{Theorem}
\newtheorem{definition}{Definition}
\newtheorem{corollary}{Corollary}

\newtheorem{lemma}{Lemma}
\newtheorem{proposition}{Proposition}
\newtheorem{conjecture}{Conjecture}
\newtheorem{observation}{Observation}
\newtheorem{remark}{Remark}

\title{Collision Detection is Instance $\widetilde{\mathrm{O}}$ptimal\\
Under the Birthday Threshold}
\author{Omri Ben-Eliezer\thanks{Technion -- Israel Institute of Technology}
  \and Tomer Grossman\thanks{Weizmann Institute}
  \and V\'aclav Rozho\v{n}\thanks{Charles University}
  \and Jakub T\v{e}tek}
\date{}

\begin{document}

\maketitle

\begin{abstract}
Can structural knowledge about a hash function help accelerate the (black box) detection of collisions in it?
This question is fundamental to cryptography theory given the importance of collision-resistant hash functions, and in this paper we tackle it from the angle of \emph{instance optimality}, an ultimate notion of beyond worst case algorithm analysis that has gained significant traction in recent years. 
Instance optimality asks for a single algorithm that, on every input,
performs nearly as well as the best correct algorithm that ``knows the structure'' of that specific
input. Here we measure algorithms by the number of queries they make to the hash function $f \colon [n] \to [n]$, and we say that an algorithm ``knows the structure'' of the input if, in addition to query access to $f$, the structure-aware algorithm has free access to an unlabeled copy $\pi^{-1} \circ f \circ \pi$ of $f$, for an unknown permutation $\pi$ on $[n]$. 

We prove the existence of an (almost) instance-optimal algorithm for collision detection in the regime most interesting from a cryptographic perspective: among functions where finding a collision takes significantly less than $\sqrt{n}$ queries.
Specifically, we prove the existence of a single algorithm $A$ that, for any input $f$ in which a structure-aware algorithm can find a collision using $q \leq O(\sqrt{n/\log n})$ queries in expectation, $A$ can find a collision in at most $O(q \log n)$ queries. The $O(\log n)$ multiplicative overhead is tight, matching a lower bound of Ben-Eliezer, Grossman, and Naor [ICALP'25], and partially resolving their main open question. Our result implies, in particular, that it is impossible for a cryptographic designer to plant purely structural backdoors for collision finding (for this unlabeled notion of structure): whatever collisions the designer's secret
knowledge can reach, the public reaches at a multiplicative price of $O(\log n)$,
without ever knowing whether a backdoor was planted.
\end{abstract}

\section{Introduction}
\label{sec:intro}

Suppose that we are given classical, black-box query access to a function
$f\colon [n] \to [n]$, and want to find a collision: distinct
$x,y\in [n]$ such that $f(x)=f(y)$. For a uniformly random function,
$\Theta(\sqrt n)$ queries are necessary and sufficient. This is known as the
birthday bound; a family $\mathcal{H}$ of hash functions is considered collision-resistant if one cannot find a collision in a function sampled from $\mathcal{H}$ much faster than the birthday bound. The design of collision-resistant families of hash functions is a cornerstone task in cryptography \cite{Damgard87,RogawayShrimpton04}.

Now suppose that an algorithm is given -- for free -- the exact structure of $f$ up to
a relabeling of its vertices, in addition to simple black box query access to $f$. Such an algorithm knows exactly the arrangement and lengths of the
paths and cycles in $f$, but does not know which label corresponds to each location in this unlabeled picture. Can this information lead to a substantially faster
collision-finding algorithm?
This is the question
of \emph{instance optimality}, a notion introduced by Fagin, Lotem, and
Naor \cite{FaginLN03}. An algorithm is $\alpha$-instance optimal if, on
every input, its cost is at most $\alpha$ times that of any correct
algorithm on the same input. Where such a guarantee is attainable, it is
the strongest kind of beyond-worst-case analysis: no other algorithm can
improve on it by more than the factor $\alpha$, even on a restricted
family of inputs.

In the query model, the comparison must be formulated with some care. An
algorithm designed with full knowledge of the input can output the answer
without making a query, and no algorithm can compete with that benchmark.
The benchmark should therefore know a great deal about the input, but not
everything. In \emph{unlabeled instance optimality}, it knows the input up
to a relabeling of the ground set. It may exploit the structure of the
input, but it does not know the labels through which that structure can be
queried. This model was introduced by Grossman, Komargodski, and Naor
\cite{GrossmanKN20} in the decision-tree setting and was subsequently
studied for search problems over functions and graphs by Ben-Eliezer,
Grossman, and Naor \cite{BGN24}.

For collision detection, \cite{BGN24} showed that structural
information -- in the form of an unlabeled copy, as described above -- can help by a factor of $\Omega(\log n)$. They conjectured
that this is the largest possible advantage.
We prove the conjecture for functions whose structure-aware query
complexity is at most $c\sqrt{n/\log n}$, for an absolute constant
$c>0$. 

\subsection{Unlabeled instance optimality}
\label{subsec:instance-opt}

We first describe the comparison formally. Throughout the paper,
\emph{relabeling} means relabeling the vertices of the functional graph.\footnote{By \emph{functional graph} of a function $f \colon D \to D$, we mean the directed graph with out-degree $1$, where the single out-edge from $x$ is directed toward $f(x)$.}
For a permutation $\pi \in S_n$, write
\[
        f^\pi:=\pi^{-1}\circ f\circ\pi.
\]
Thus $f$ and $f^\pi$ have isomorphic functional graphs. An algorithm may
be tailored to the class $\{f^\pi:\pi\in S_n\}$, but it does not know
which labeling it is queried on.

We use the framework of unlabeled instance optimality
\cite{GrossmanKN20,BGN24}. All algorithms are Las Vegas: they always
return the correct answer, while their number of queries is random. For
an algorithm $A$, let $\Queries_A(f)$ denote its expected number of
queries on $f$.

\begin{definition}[Unlabeled instance optimality; \cite{GrossmanKN20},
see also \cite{BGN24}]
\label{def:inst_opt_inf}
Let $\mathcal F$ be a class of functions $f\colon[n]\to[n]$. A Las
Vegas algorithm $A$ for collision detection is
$\alpha(n)$-\emph{instance optimal with respect to $\mathcal F$} if,
for every $f\in\mathcal F$ and every Las Vegas algorithm $A'$ for
collision detection,
\[
   \Queries_A(f)
   \le
   \alpha(n)\cdot
   \max_{\pi\in S_n}\Queries_{A'}(f^\pi).
\]
Here $A'$ must be correct on every function $[n]\to[n]$, including
functions outside $\mathcal F$. Collision detection is
$\alpha(n)$-instance optimal with respect to $\mathcal F$ if it admits
such an algorithm.
\end{definition}

The order of quantifiers is important. The algorithm $A$ is fixed. Its
competitor $A'$ may be chosen specifically for the structure of $f$, but
it must remain correct when that structural hint is wrong. The maximum
over $\pi$ charges $A'$ for the worst labeling of the given structure.

Ben-Eliezer, Grossman, and Naor proved that collision detection is not
$o(\log n)$-instance optimal and conjectured that the logarithmic gap is
tight.

\begin{conjecture}[{\cite[Conjecture 1.5]{BGN24}}]
\label{conjecture:BGN24}
Collision detection admits an $O(\log n)$-instance optimal algorithm.
\end{conjecture}

It was specifically conjectured in \cite{BGN24} that the $O(\log n)$ instance optimality is attained by the \emph{all-scales algorithm}, defined as follows. Maintain one walker for
each $i=0,1,\ldots,\lceil\log n\rceil$. A walker at scale $i$ starts
at a uniformly chosen vertex and follows $f$ for at most $2^i$ steps,
or until its walk repeats a vertex. All walkers run in parallel, and the
algorithm stops when the queried edges contain two edges with the same
head -- i.e., a collision. We analyze a slight variant of this algorithm, in which the walkers use a common random
ordering of the starting vertices and run in epochs of doubling length.
The full definition appears in \cref{alg:allscales}, and the difference
from the original variant is discussed in \cref{rem:variant}.

\subsection{Our result}
\label{subsec:main-result}

For $f\colon[n]\to[n]$, define its structure-aware query complexity by
\[
   \OPT(f):=
   \inf_{A'}\max_{\pi\in S_n}\Queries_{A'}(f^\pi),
\]
where the infimum ranges over all Las Vegas algorithms for collision
detection. For $q=q(n)$, let
\[
   \mathcal F_q:=\{f\colon[n]\to[n]:\OPT(f)\le q\}.
\]

\begin{theorem}
\label{thm:main}
There is an absolute constant $c>0$ such that the all-scales algorithm
 (\cref{alg:allscales}) is $O(\log n)$-instance optimal with respect to
$\Feasy$.
\end{theorem}

Equivalently, for every $f\in\Feasy$ and every Las Vegas algorithm
$A'$ for collision detection,
\[
   \Queries_{\Aall}(f)
   \le O(\log n)\cdot
   \max_{\pi\in S_n}\Queries_{A'}(f^\pi).
\]
In the construction of \cite{BGN24} showing an $\Omega(\log n)$ separation, the expected query
complexity of the structure-aware algorithm is $n^\gamma$ for a fixed $\gamma<1/2$. It therefore lies in
$\Feasy$ for all sufficiently large $n$. Thus, the logarithmic factor in
\cref{thm:main} is tight up to a multiplicative constant.

We prove \cref{thm:main} in two stages. The coupling argument first gives
the same conclusion for $\Feasywarm$; this is \cref{thm:warmup}. A
second probability estimate removes a factor $\sqrt{\log n}$ from the
threshold. The coupling and all deterministic parts of the proof are the
same in both stages.

\paragraph{Cryptographic interpretation.} Collision resistance asks that no efficient adversary find a collision in a function drawn from a public family~\cite{Damgard87,RogawayShrimpton04}. For a random function the birthday bound $\Theta(\sqrt n)$ is optimal, and rho-style random walks attain it with small memory~\cite{Pollard75,vanOorschotW99}; the interesting case is a function that \emph{does} have structure---an atypical profile of cycle lengths, tail lengths or in-degrees---whether by accident of design or because a designer planted it. Relabeling captures exactly what such a designer knows and the public does not: the family $\{f^\pi : \pi\in S_n\}$ is the set of all functions with the shape of $f$, and an adversary who knows the shape but not the secret labeling $\pi$ is precisely a competitor $A'$ charged $\max_\pi Q_{A'}(f^\pi)$, as in Definition~\ref{def:inst_opt_inf}. In this language, Theorem~\ref{thm:main} says the following. If some shape-aware algorithm finds a collision in $f$ with $q\le c\sqrt{n/\log n}$ expected queries---that is, if the structural weakness is worth more than a $\sqrt{\log n}$ factor over the birthday bound---then the single public algorithm $A_{\mathrm{all}}$, given no information about $f$, finds one in $O(q\log n)$ expected queries; and since the behavior of $A_{\mathrm{all}}$ is invariant under relabeling, the same bound holds for every $f^\pi$ simultaneously. Purely structural knowledge of a hash function therefore cannot serve as a backdoor in the query model: whatever collisions the shape lets the designer reach below the birthday threshold, the public reaches at a multiplicative cost of $O(\log n)$, without knowing whether a weakness exists or at which scale it lives, and by the lower bound of~\cite{BGN24} no universal algorithm can do better. This complements the usual analysis of generic collision search, in which the function is modeled as a random oracle~\cite{BellareK04,FlajoletO90,vanOorschotW99}, and is in the spirit of Rogaway's ``human ignorance'' treatment of unkeyed hash functions~\cite{Rogaway06}, in that the guarantee is per instance and is witnessed by an explicit, fixed algorithm. We stress the limits of the statement: it concerns query complexity only; it bounds the advantage of structural (unlabeled) knowledge, not of a trapdoor tied to the labeling; and it applies below the birthday threshold, which is the only regime in which a speedup is meaningful.

\subsection{Technical overview}
\label{subsec:overview}

We next describe the main ideas of the proof, which has three parts. In the first part, we show that any algorithm for collision detection can be translated to a normal form of that algorithm, which follows a certain ``oldest first'' principle that will make our analysis easier later on. In the second part, we construct a probabilistic coupling of any normal-form algorithm with an instance of the all-scales algorithm (which is the candidate instance optimal algorithm), trying to match events where the algorithms decide to sample a new, previously unseen vertex. Finally, we bound the probability
that the coupling fails.

\paragraph{The normal form.}
Fix $f$ and a competing Las Vegas algorithm $A'$. We first apply a
uniformly random relabeling before simulating $A'$. The resulting
algorithm is label-oblivious, and its expected complexity on $f$ is at
most $\max_\pi\Queries_{A'}(f^\pi)$; see \cref{lem:relabel}.
Conditioned on the unlabeled transcript of such an algorithm, every
consistent embedding of that transcript in the functional graph of $f$
is equally likely (\cref{lem:uniform}).

Any algorithm can be seen as maintaining a set of walks (and cycles, which are walks that can no longer be extended). In each step, as long as no collision has been found, the algorithm can choose to either extend a path or sample a fresh, never-seen-before vertex. Suppose that no two walks merge through the algorithm's run (this no-merging requirement generally only holds below the birthday bound, and is the main reason our proof only applies for $f \in \Feasy$). Using a virtual simulation argument, we argue that if the algorithm chooses to extend a path of length $\ell > 0$ at any step, then it might as well pick the \emph{oldest} (i.e., earliest-born) path among its functional graph, and extend this path. We call this an ``oldest-first'' principle (\cref{lem:standard}). 

Thus, any competing algorithm can be viewed as a
scheduling rule that chooses in each round whether to start a new walk, and if not, which walk length to extend. We consider algorithms operating according to the oldest-first principle as normal-form algorithms.

\paragraph{The shared-root coupling.}
For convenience in the analysis, we may assume that all randomness in the process is generated in advance (but not known to the algorithms). In particular, let $R=(r_1,\ldots,r_n)$ be a uniformly random ordering of the
vertices. We can view any fresh-vertex
query of the competing algorithm as the operation of scanning $R$ until reaching (and picking) the first
vertex not already embedded in its transcript. Crucially, we will think of the all-scales algorithm as using the same random sequence; the two algorithms here are \emph{coupled}. 

Choose a normal-form competitor $A$ with expected query complexity
$q^*\le2\OPT(f)$. Let $t$ be the smallest number for which $A$ finds a
collision within $t$ queries with probability at least $9/10$.
Markov's inequality gives
\[
        t\le10q^*\le20\OPT(f).
\]

As mentioned, let $R=(r_1,\ldots,r_n)$ be the uniformly random ordering of the
vertices dictating the order of fresh vertices queried by the competing algorithm. Every walker (at every scale) of all-scales will use this ordering as well.

Assuming a ``good event'' defined
below, the scan of $R$ by the competing algorithm will never see an already-queried node, and so this algorithm will not make any skips among the first $t$ positions. Consequently
the $j$-th root of $A$ and the $j$-th root processed by every walker of all-scales are
the same vertex $r_j$, for all $1 \leq j \leq t$.

\paragraph{Boxes and the good event.}
For a vertex $u$, let $\rho(u)$ be the number of steps before the walk
from $u$ first repeats a vertex, and let
$\type(u)=\lfloor\log\rho(u)\rfloor$. For every scale $i$, define a
box $B_i$ containing the first roughly $t/2^i$ positions $j$ for which
$\type(r_j)\ge i$. The oldest-first rule gives a budget bound: within
$t$ queries, the competitor can explore to depth $2^i$ only from roots
whose positions lie in $B_i$. Before pushing a later walk that far, it
must first push every older open walk of the same length, unless an older
walk has already closed and paid its full cost.

For every $j\in B_i$, ``protect'' the first $2^{i+1}$ vertices of the walk
from $r_j$. The number of protected vertex occurrences is
\[
   \sum_i O(t/2^i)\cdot O(2^i)=O(t\log t).
\]
The \emph{good event} $G$ is that none of the other roots $r_1,\ldots,r_t$ is one
of these protected vertices. On $G$, no root is skipped in the coupling,
and the walks explored by $A$ remain disjoint until a collision is found.
This is formalized by the exploration lemma, \cref{lem:exploration}.

\paragraph{Domination on the good event.}
Condition on the good event $G$, and suppose that $A$ discovers the edge at depth
$\ell$ from $r_j$ within $t'$ queries. Set
$m=\lceil\log\ell\rceil$, so that $\ell\le2^m<2\ell$. Before
reaching $r_j$, walker $W_m$ spends
\[
       \min\{2^m,\rho(r_{j'})\}
\]
rounds on each earlier root $r_{j'}$. At the corresponding moment, the
oldest-first competitor has spent at least
$\min\{\ell,\rho(r_{j'})\}$ queries on that walk. Since
\[
   \min\{2^m,x\}\le2\min\{\ell,x\}
   \qquad\text{for every }x\ge1,
\]
$W_m$ reaches the edge by round $2t'$. Hence, if $A$ finds a collision
within $t$ queries, the coupled epoch contains both edges of a collision
witness within $2t$ rounds (\cref{lem:domination,cor:transfer}). Each
round uses at most $O(\log n)$ oracle queries.

\paragraph{A first probability bound on the good event.}
We now wish to bound $\Pr[\neg G]$. The good event fails if some root
$r_{j}$ can be reached from some vertex $r_k$ of type $i$ within no more than $2^{i+1}$ steps. There are $O(t \log t)$ such target vertices, so if they were
fixed in advance, a union bound over the $t$ roots would give a failure
probability of $O(t^2 \log t / n)$. The issue however is that they are not fixed: membership in the
boxes is determined by the survival times $\rho(r_1), \dots, \rho(r_t)$,
so the targets are correlated with the very roots we test. To bound the probability of the good event, we would like to condition on the type
$\lfloor \log \rho(r_j) \rfloor$ of every root. Given the types, the
boxes become deterministic; the amount of vertices ``at risk'' from type-$i$ roots is only a function of the number of type-$i$ roots. The price of the conditioning is that
small type classes may be overpopulated by the sample, and by Markov's inequality, controlling
this simultaneously for all $\Theta(\log n)$ classes loses a
$\Theta(\log n)$ factor. Thus, the failure probability is $O(t^2 \log n \log t / n)$. This analysis suffices to prove the main theorem for all $f \in \mathcal{F}_{c\sqrt{n} / \log n}$, i.e., it is a multiplicative factor of $\sqrt{\log n}$ away from the actual statement of the theorem. 

\paragraph{A refined probability analysis.}
We next show how to shave a $\log n$ factor from the failure probability. This allows us to extend the range in which our results work by a $\sqrt{\log n}$ factor.

The idea is to reveal the randomness more carefully; why not expose the roots one at a time?  If the walk from a root $u$ is
already fixed, then a newly exposed uniform root hits its first $L$
vertices with probability at most $L/(n-t+1)$.  The reverse statement however is
false: for a fixed vertex $v$, there may be many---even all---starting
vertices whose first $L$ steps reach $v$.  A sequential exposure therefore
controls a new root entering an old walk, but not a new walk entering an
old root.  Every failure of $G$ is such a directed pair, and no single order of revealing randomness puts the random endpoint second in every pair.

To circumvent the assymetry described above, we separately bound the probability for each root $r_{j}$ to be reachable by the set of all other roots, without attempting to bound the analogous event (of $r_{j}$ reaching another root). To do so, after
revealing the other roots, delete position $j$ and, at each scale, take
the first $\lceil t/2^i\rceil$ remaining $i$-long positions. These
surrogate boxes contain every original box member other than $j$: deleting
one position can only move the other positions forward in rank.  Their
total protected region has size $O(t\log t)$ and is fixed before
$r_{j}$ is exposed.  Conditional on the other roots, $r_{j}$ is uniform
among the $n-t+1$ unused vertices.  Hence a union bound over $j$ gives the desired bound
\[
       \Pr[\neg G]
       =O\!\left(\frac{t^2\log t}{n}\right).
\]

\paragraph{From one epoch to expectation.}
So far we have seen that if any (structure-aware) algorithm $A$ can find a collision with probability $p$ after $q$ queries on $f$, then the all-scales algorithm can do the same with probability $p' \geq p$ after $O(q \log n)$ queries. By Markov inequality, $A$ finds a collision after $2Q_A(f)$ queries with probability at least $1/2$. Thus, all-scales will find a collision with such probability after $O(Q_A(f) \log n)$ queries. To achieve the same bound in expectation, we run all-scales again and again (each time with fresh randomness) and use standard expectation analysis of geometric random variables to bound the expected Las Vegas query complexity.

\subsection{Related work}
\label{subsec:related}

\paragraph{Instance optimality and unlabeled information.}
Fagin, Lotem, and Naor \cite{FaginLN03} introduced instance optimality in
their study of aggregation algorithms: an algorithm is compared, on every
instance, with every correct algorithm from a prescribed class. Grossman,
Komargodski, and Naor \cite{GrossmanKN20} initiated a systematic study of
this notion in the decision-tree model and introduced \emph{unlabeled
certificates}, which retain the structure of an input while hiding the
names of its coordinates. Subsequent work developed the instance
complexity of Boolean functions \cite{HsiangL23}; a related
beyond-worst-case benchmark is min-entropic optimality \cite{ArnonG21}.

The work closest to ours is that of Ben-Eliezer, Grossman, and Naor
\cite{BGN24}, who studied unlabeled instance optimality for detecting
substructures in functions and graphs. They proved the lower bound that
underlies \cref{conjecture:BGN24}, together with positive results for claw
detection in the easy regime and for collision detection when backward
queries are allowed. The present paper resolves a substantial range of
the forward-query problem (which is much more natural for cryptographic applications), where the algorithm sees only values of the form
$f(x)$ but cannot query the set of sources $f^{-1}(x)$ for a given $x$. 

\paragraph{Instance-sensitive and universal algorithms.}
Guarantees tailored to an individual input have appeared in geometric
algorithms \cite{AfshaniBC17}, adaptive set operations and sorting
\cite{DemaineLM00,BaranD04,sorting1,sorting2,Narayanan2024}, distribution
testing and learning
\cite{ValiantV16,ValiantV17,HaoOSW18,HaoOrlitsky2020}, best-arm
identification \cite{chen17b,Li2022}, and sublinear graph algorithms
\cite{Gonen2011,tetek2022}. A particularly relevant precedent comes from
sorting under partial information: the comparison bound is governed by
the combinatorial structure, or entropy, of the admissible orders
\cite{Fredman76,KahnK95}. Self-improving algorithms instead learn an input
distribution and approach the best expected running time for that
distribution \cite{AilonCCLMS11}. Under the name \emph{universal
optimality}, analogous goals have been pursued in distributed computing
\cite{HWZ21} and, more recently, for shortest paths
\cite{Dijkstra2024,haeupler2024bidirectional}.

\paragraph{Symmetry in decision-tree complexity.}
Our formulation is also related to the broad literature on decision-tree
and query complexity; see the survey of Buhrman and de Wolf
\cite{BuhrmanW02} and Yao's minimax framework \cite{Yao77}. Symmetry has
long played a central role in lower bounds, most notably for graph
properties \cite{RivestV76,KahnSS84}. There is, however, an important
difference in how symmetry enters here. Classical evasiveness results use
invariance to prove a worst-case lower bound for evaluating one fixed
property. In unlabeled instance optimality, symmetry defines the
\emph{comparison class}: the competing algorithm may be tailored to the
orbit of the particular input, although it must still be correct on every
input.

\paragraph{Generic collision search.}
Random-walk methods for finding collisions go back to Pollard's rho method
\cite{Pollard75,Pollard78}. Important refinements include Brent's
cycle-detection procedure \cite{Brent80} and the parallel collision search
of van Oorschot and Wiener \cite{vanOorschotW99}. The probabilistic
behavior of such walks is closely connected to random-mapping statistics
\cite{FlajoletO90}. Time--memory tradeoffs for function inversion provide
another influential generic paradigm \cite{Hellman80,FiatN00}, while
Bellare and Kohno \cite{BellareK04} quantify how nonuniform hash outputs
affect birthday attacks. We note that these works generally do not consider the per-instance complexity, as opposed to this paper.

\paragraph{Cryptographic and quantum collision notions.}
Collision resistance has also motivated refined definitions and attacks
for hash functions, including keyless formulations \cite{Rogaway06},
multicollision attacks on iterated hashing \cite{Joux04}, and constructions
resilient to many collisions \cite{KomargodskiNY18}. In the quantum query
model, collision finding and element distinctness admit different
complexity bounds from their classical counterparts
\cite{BrassardHT98,AaronsonS04}. These lines are conceptually adjacent but
not directly comparable to our result: our model is classical and Las
Vegas, the function is an arbitrary fixed input rather than a random
oracle or a cryptographic construction, and the benchmark is the
instance-specific advantage supplied by unlabeled structural advice.

\subsection{Organization}
\label{subsec:intro-open}

\Cref{sec:prelim} defines the model and the all-scales algorithm.
\Cref{sec:standard} proves the oldest-first normal form, with full proofs
deferred to Appendix~\ref{app:reduction}. \Cref{sec:proof} develops the
coupling and proves the warm-up theorem. \Cref{sec:sharpening} proves the
sharper probability bound and derives \cref{thm:main}.
\Cref{sec:questions} concludes with open problems.

\providecommand{\revisioncolor}{}

\section{Preliminaries}
\label{sec:prelim}

Throughout the paper, $\log$ denotes the base-$2$ logarithm,
$[n] = \{1, \dots, n\}$, and $n$ is assumed to be larger than a
sufficiently large absolute constant. (For smaller $n$, all statements
hold trivially by adjusting the constant in the $O(\log n)$ factor; see
the proofs of \cref{thm:warmup,thm:main}.) We identify a function
$f \colon [n] \to [n]$ with its \emph{functional graph}: the directed
graph on vertex set $[n]$ with an edge $u \to f(u)$ for every $u$, so
that every vertex has out-degree exactly one.

\begin{definition}[Collision detection]
\label{def:collision}
Given query access to a function $f \colon [n] \to [n]$, the goal is to
either find a collision --- two vertices $u_1 \neq u_2$ where $f(u_1) = f(u_2)$ --- or to answer that no collision exists.
\end{definition}

We use the classical black-box query model. A query specifies
$x\in[n]$ and returns $f(x)$. An algorithm is \emph{Las Vegas} if it
always returns the correct answer; only its number of queries is random.
For a Las Vegas algorithm $A$, let $\Queries_A(f)$ denote the expected
number of queries on $f$ until it terminates (by finding a collision or stating that no collision exists). We do not consider other computational or memory costs in this paper.

After the symmetrization in \cref{sec:standard}, every useful query has
one of two forms. It either queries the out-neighbor of a vertex already
seen, thereby extending a walk, or queries a label not seen before,
thereby starting a new walk and discovering its first edge. In the latter
case the corresponding vertex of the hidden functional graph is uniform
among the vertices not exposed so far. We shall generate these choices
using a uniformly random permutation of $[n]$; see
\cref{subsec:setup}. This is only a description of the randomness of an
ordinary black-box algorithm, not an additional oracle operation.

\paragraph{Walks, $\rho$-values and types.}
For a vertex $u$ and an integer $\ell \ge 0$ we write
\[
u +_f \ell := f^{(\ell)}(u)
\]
for the $\ell$-th iterate of $f$ started at $u$; the \emph{walk from $u$}
is the sequence $u, u+_f 1, u+_f 2, \dots$. 
Since every vertex has out-degree
one, the walk eventually revisits a vertex; define
\[
\rho(u) := \min\bigl\{ m \ge 1 : u +_f m \in \{u, u+_f 1, \dots, u+_f (m-1)\} \bigr\}
\;\le\; n.
\]
The first $\rho(u)$ steps of the walk trace a ``$\rho$ shape'': a simple
path (the \emph{tail}) of some length $\sigma(u) \ge 0$, followed by a
cycle of length $\rho(u) - \sigma(u)$, with $u +_f \rho(u) = u +_f \sigma(u)$.
If $\sigma(u) = 0$ the walk is a pure cycle through $u$ and witnesses no
collision; if $\sigma(u) \ge 1$ then the vertex $u +_f \sigma(u)$ has two
distinct in-edges among the walk's edges, witnessing a collision.
Finally, define the \emph{type} of $u$ as
\[
\type(u) := \lfloor \log \rho(u) \rfloor \in \{0, 1, \dots, \lfloor \log n \rfloor\}.
\]

The following two observations describe what a Las Vegas algorithm must
see before it can stop. Together, they say that the transcript of any
correct algorithm must contain an explicit witness.

\begin{observation}
\label{obs:collision-witness}
A Las Vegas algorithm can output ``collision'' only when its transcript
contains two distinct discovered edges pointing into the same vertex.
\end{observation}

\begin{proof}
Suppose the discovered edges are pairwise into distinct vertices. Then
they form a partial injection on $[n]$, which extends to a permutation
$g$ of $[n]$. The transcript is consistent with the input being $g$, and
$g$ has no collision. Hence
on input $g$ the algorithm would, with positive probability, produce this
very transcript and err --- contradicting the Las Vegas property.
\end{proof}

\begin{observation}
\label{obs:no-collision-witness}
A Las Vegas algorithm can output ``no collision'' only when its
transcript contains the out-edges of all $n$ vertices. Consequently,
every Las Vegas algorithm makes at least $n$ queries on every
collision-free input, and every $f$ with $\OPT(f) < n$ contains a
collision (where
$\OPT(f) := \inf_{A'} \max_{\pi} \Queries_{A'}(f^\pi)$ as in
\cref{subsec:main-result}).
\end{observation}

\begin{proof}
Suppose the out-edge of some vertex $u$ is not in the transcript. If the
transcript contains at least one edge $w \to v$, modify the input to $f'$
with $f'(u) := v$ and $f' = f$ elsewhere; then $f'$ has a collision
($u \neq w$ map to $v$) and is consistent with the transcript, which
again has positive probability under $f'$. If the transcript contains no
edge at all, it is consistent with a constant function. Either way,
answering ``no collision'' errs with positive probability on some input.
\end{proof}

\subsection{The all-scales algorithm}
\label{subsec:allscales}

We now define the algorithm we analyze, a variant of the all-scales
algorithm of \cite{BGN24}. It maintains one \emph{walker} $W_i$ for
each scale $i=0,1,\ldots,L:=\lceil\log n\rceil$. During an epoch, all
walkers use the same uniformly random ordering
$r_1,r_2,\ldots,r_n$ of the vertices. Epoch lengths double, and the
algorithm keeps all edges discovered in earlier epochs.

\begin{definition}[All-scales algorithm, shared-stream variant]
\label{def:allscales}
The algorithm $\Aall$ maintains a global \emph{memory} $M$ of discovered
edges $u \to f(u)$, initially empty, and proceeds in epochs
$e = 1, 2, 3, \dots$ At the start of each epoch, a fresh uniformly
random permutation $r_1,r_2,\ldots,r_n$ is fixed (lazily), and
every walker is reset to the beginning of the stream; the memory $M$ is
kept. The epoch lasts $2^e$ rounds. In each round, every walker performs
one unit of work:
\begin{itemize}
\item if the walker is \emph{idle}, it starts a new walk at the next
vertex of the permutation that it has not yet processed and immediately
takes the first step of that walk; if it has processed all $n$ vertices,
it does nothing;
\item otherwise, it advances its current walk by one edge --- if the edge
is in $M$ this is a free lookup, and otherwise it is one query, whose
answer is added to $M$.
\end{itemize}
Walker $W_i$ ends its current walk (and becomes idle) when the walk has
made $2^i$ steps, or when the walk revisits one of its own vertices,
whichever comes first. At the end of every round, if two distinct edges
of $M$ point into the same vertex, the algorithm reports that collision
and terminates; if $M$ contains the out-edges of all $n$ vertices, the
algorithm reports ``no collision'' and terminates.
\end{definition}

\ifdefined\theoryversion
\begin{algorithm}[H]
\else
\begin{algorithm}
\fi
\caption{The all-scales algorithm $\Aall$ (shared-stream, doubling-epochs variant)}
\label{alg:allscales}
\begin{algorithmic}[1]
\State $M \gets \emptyset$ \Comment{global memory: set of discovered edges $u \to f(u)$}
\For{epoch $e = 1, 2, 3, \dots$}
  \State fix a fresh uniformly random permutation $r_1,r_2,\ldots,r_n$ (sampled lazily)
  \State reset walkers $W_0, W_1, \dots, W_{\lceil \log n \rceil}$; each is idle, at stream position $0$
  \For{$2^e$ rounds}
    \For{\textbf{all} walkers $W_i$} \Comment{one unit of work per walker per round}
      \If{$W_i$ is idle}
        \If{$W_i$ has not yet processed all $n$ vertices}
          \State $W_i$ starts at its next vertex $r_j$ and takes the first step
        \EndIf
      \Else
        \State $W_i$ advances its walk by one edge, via $M$ if known, else by one query (added to $M$)
      \EndIf
      \If{$W_i$'s walk has made $2^i$ steps or revisited one of its own vertices}
        \State declare $W_i$ idle
      \EndIf
    \EndFor
    \If{two distinct edges in $M$ point into the same vertex $v$}
      \State \Return the collision at $v$
    \EndIf
    \If{$M$ contains the out-edges of all $n$ vertices}
      \State \Return ``no collision''
    \EndIf
  \EndFor
\EndFor
\end{algorithmic}
\end{algorithm}

We use three accounting conventions. First, choosing the next starting
vertex costs no query; the first step from it is an ordinary query unless
its edge is already in memory. A lookup of a known edge is free, but it
still uses the walker's unit for that round. Thus trajectories and
timings depend only on $f$ and the permutation, not on the current
memory. Second, each round costs at most
$L+1\le\log n+2$ queries. Third, a walker stops only on its own step cap
or on a self-repeat of its
current walk; if its walk merges into territory explored by another
walker (or by itself in an earlier walk), it keeps walking through known
edges. Collisions are detected not by the walkers but by the memory: as
soon as $M$ holds two in-edges of the same vertex, from whatever source,
the collision is reported. In particular, when a walk self-repeats at a
vertex other than its own starting point, the repeat vertex has two
in-edges in $M$ and the collision is reported in the same round.

\begin{remark}[Relation to the algorithm of \cite{BGN24}]
\label{rem:variant}
The all-scales algorithm described in \cite{BGN24} lets the walkers
choose independent starting vertices and does not restart. Our variant
uses a shared random ordering and doubling epochs. The common ordering
permits the coupling with an arbitrary competing algorithm; fresh
orderings make the long epochs independent trials. We expect the
original variant to satisfy the same guarantee, but our proof does not
show this; see
\cref{sec:questions}.
\end{remark}

\begin{remark}[Termination]
\label{rem:termination}
$\Aall$ is Las Vegas. Every reported collision is certified by two
queried edges, and ``no collision'' is reported only when all $n$ edges
are known. Moreover, $W_0$ queries the edge of one new vertex in each
round. It therefore exposes all edges during any epoch of at least $n$
rounds, at which point the algorithm terminates with the correct answer.
In particular, the doubling schedule gives finite expected query
complexity on every input.
\end{remark}

\section{Reduction to label-oblivious, oldest-first algorithms}
\label{sec:standard}

Our benchmark, $\max_\pi \Queries_{A'}(f^\pi)$, is indifferent to vertex
labels, and this section makes the competing algorithm indifferent to
them too. We normalize an arbitrary Las Vegas algorithm $A'$ in three
steps: we make it label-oblivious (\cref{lem:relabel}), we discard
redundant queries (\cref{obs:normal-ops}), and we make its scheduling
oldest-first (\cref{lem:standard}). The key structural fact enabling the
last step is the uniform-embedding lemma (\cref{lem:uniform}). The
arguments are symmetrization steps of a standard flavor (compare the
reductions with unlabeled certificates in \cite{GrossmanKN20,BGN24}); we
state the lemmas here, explain the ideas, and defer the full proofs to
Appendix~\ref{app:reduction}.

\paragraph{Transcripts.}
The \emph{transcript} of an algorithm at a given time consists of the
sequence of operations performed so far together with their answers. The
\emph{unlabeled transcript} $H$ is the transcript with the vertex names
replaced by abstract placeholders: it records, for each operation,
whether it was a fresh-label query or an out-neighbor query of a specific
placeholder, and whether the answer was a new placeholder or coincided
with an existing one (and which). Thus $H$ is precisely the isomorphism
type of the explored partial graph, together with the exploration
history. An algorithm is \emph{label-oblivious} if its next operation
(and its final answer) is a randomized function of the unlabeled
transcript only.

\begin{lemma}[Relabeling]
\label{lem:relabel}
For every Las Vegas algorithm $A'$ for collision detection there is a
label-oblivious Las Vegas algorithm $A''$ such that for every $f$,
\[
\Queries_{A''}(f) \;\le\; \max_{\pi \in S_n} \Queries_{A'}(f^{\pi}).
\]
\end{lemma}

The construction is simple: $A''$ applies a uniformly random
relabeling of its own before running $A'$, so the labels $A'$ sees carry
no information beyond their coincidence pattern; the cost of $A''$ on $f$
is the average cost of $A'$ over relabelings, which is at most the
maximum.

\begin{observation}
\label{obs:normal-ops}
We may assume without loss of generality that each query is either
(i) a query at a label not seen before, called \emph{starting a root},
or (ii) a query at a seen vertex whose out-edge is not yet known. We may
also assume that the algorithm halts as soon as its transcript contains a
collision witness. Before such a witness is found, the discovered edges
form vertex-disjoint directed paths and cycles. A path component may
contain more than one root: this happens when one path enters the first
vertex of another path, which need not yet have a known in-edge. This is
the only way two components can merge without producing a collision
witness.
\end{observation}

\begin{proof}
Re-querying a known out-edge gives no information and can be skipped. A
query at a seen vertex with unknown out-edge extends a path component; a
query at an unseen label starts one. As long as no vertex has two known
in-edges, every component of the known-edge graph is a directed path or
cycle. An edge entering an interior vertex creates a collision witness.
An edge entering the first vertex of a path merely concatenates two
paths, because that vertex may have no known in-edge. Continuing after a
collision witness is unnecessary by \cref{obs:collision-witness}.
\end{proof}

\paragraph{Uniform embeddings.}
Fix the input $f$ and a label-oblivious algorithm $A$. At any point of
the execution, the unlabeled transcript $H$ comes with an
\emph{embedding}: the injective map $e$ sending each placeholder of $H$
to the concrete vertex it stands for. Call an injective map $e$ from the
placeholders of $H$ to $[n]$ \emph{consistent} if for every explored edge
$(x, y)$ of $H$ we have $f(e(x)) = e(y)$, and write $E_H$ for the set of
consistent maps. Note that consistency is exactly what the transcript
reveals: each answer either creates a new placeholder (an explored edge
to a new vertex) or is recorded as a coincidence with an existing
placeholder (an explored edge to it), and distinct placeholders always
denote distinct vertices.

The next lemma formalizes a symmetry: a label-oblivious algorithm knows
the shape of what it has explored and nothing more, so all consistent
placements of that shape remain equally likely. The proof is by induction
on the queries. Assigning an unseen label treats all unused vertices
alike, and revealing its out-neighbor partitions the consistent
placements according to the observed coincidence pattern.

\begin{lemma}[Uniform embedding]
\label{lem:uniform}
Let $A$ be label-oblivious and fix $f$. At every point of the execution,
conditioned on the unlabeled transcript $H$ (and on all of $A$'s internal
coins), the embedding is uniformly distributed on $E_H$.
\end{lemma}

\paragraph{Oldest-first algorithms.}
Call a transcript \emph{unmerged} if each path component contains one
root. On such a transcript, the length of a path is the number of edges
discovered from its root. We say that $A$ is \emph{oldest-first} if,
whenever its transcript is unmerged and it chooses to extend a path of
length $\ell$, it extends the path of that length whose root was started
first. No condition is imposed after a witness-free merge.

Two unmerged open paths of the same length are exchangeable under the
uniform embedding of \cref{lem:uniform}. There is a small issue: an
algorithm may base later decisions on the full history, including which
of these paths it extended. We therefore keep a virtual execution of the
original algorithm and a dynamic correspondence between its paths and
the actual paths. Whenever the virtual algorithm selects a path of length
$\ell$, we map that path to the oldest actual path of length $\ell$
before making the query. Exchangeability gives the correct distribution
for the answer. Appendix~\ref{app:reduction} gives the details.

\begin{lemma}[Oldest-first normalization]
\label{lem:standard}
For every label-oblivious Las Vegas algorithm $A$ there is an
oldest-first label-oblivious Las Vegas algorithm $A^\circ$ with
$\Queries_{A^\circ}(f) = \Queries_A(f)$ for every $f$.
\end{lemma}

Combining \cref{lem:relabel,lem:standard}:

\begin{corollary}
\label{cor:normalform}
For every Las Vegas algorithm $A'$ for collision detection there is an
oldest-first label-oblivious Las Vegas algorithm $A$ with
\[
\Queries_A(f) \;\le\; \max_{\pi \in S_n} \Queries_{A'}(f^{\pi})
\qquad \text{for every } f \colon [n] \to [n].
\]
\end{corollary}

In summary, every competitor may be assumed to start new roots or extend
open paths according to a label-oblivious scheduling rule. Before the
first witness-free merge, paths of the same length are extended oldest
first. The good event in \cref{subsec:boxes} will ensure that no such
merge occurs in the part of the execution used by the proof.

\section{The coupling argument}
\label{sec:proof}

This section contains the main argument. We construct a coupling between
an arbitrary normalized competitor and a single epoch of $\Aall$
(\cref{subsec:setup,subsec:boxes}), show that on a suitable good event
the epoch keeps pace with the competitor up to a factor of two
(\cref{subsec:exploration,subsec:domination}), convert the resulting
constant-probability guarantee into a bound on expected query complexity
(\cref{subsec:expectation,subsec:master}), and then bound the failure
probability of the good event by a first-moment argument
(\cref{subsec:goodevent}). Together these steps prove:

\begin{theorem}[Warm-up: threshold $\sqrt{n}/\log n$]
\label{thm:warmup}
There is an absolute constant $c_0 > 0$ such that the all-scales
algorithm (\cref{alg:allscales}) is $O(\log n)$-instance optimal with
respect to $\Feasywarm$.
\end{theorem}

The only ingredient that is sensitive to the exact threshold is the
probability bound for the good event; every other step works for any
budget. \Cref{sec:sharpening} improves that single lemma, and with it the
threshold, yielding \cref{thm:main}.

\subsection{Setup and coupling}
\label{subsec:setup}

Fix an input $f$ with $\OPT(f) < n$, where
$\OPT(f) = \inf_{A'} \max_\pi \Queries_{A'}(f^\pi)$ as in
\cref{subsec:main-result}. By \cref{obs:no-collision-witness}, $f$
contains a collision. Choose a Las Vegas algorithm $A^*$ with
$\max_\pi \Queries_{A^*}(f^\pi) \le 2\OPT(f)$, and let $A$ be its
oldest-first label-oblivious normalization from \cref{cor:normalform}, so
that
\[
q^* := \Queries_A(f) \;\le\; 2\OPT(f).
\]
Define the \emph{$9/10$-quantile budget}
\[
t = t(A) := \min\bigl\{ s \in \mathbb{N} :
\Pr[A \text{ outputs a collision within } s \text{ queries on } f] \ge 9/10 \bigr\}.
\]
By Markov's inequality applied to the number of queries (recalling that
$A$, being Las Vegas on an input with a collision, outputs a collision
upon halting),
\begin{equation}
\label{eq:tbound}
t \;\le\; 10\,q^* \;\le\; 20\,\OPT(f).
\end{equation}
Until \cref{prop:master}, the budget $t$ is arbitrary; it is compared
with the two thresholds only after the coupling has been established.

\paragraph{The coupling.}
Let $R=(r_1,r_2,\ldots,r_n)$ be a uniformly random permutation of
$[n]$. The walkers process $R$ in order. To generate a fresh-label query
of $A$, scan $R$ from the current position and take the first vertex not
already embedded in $A$'s transcript. This has exactly the required
distribution: conditioned on the transcript, it is uniform among the
unexposed vertices. On the good event defined below, no scan skips a
vertex among $r_1,\ldots,r_t$; hence the $j$-th root of $A$ is $r_j$,
which gives the desired shared-root coupling. All remaining coins of $A$
are independent of $R$.
Throughout \cref{subsec:boxes,subsec:exploration,subsec:domination}, and
in the good-event bounds of \cref{subsec:goodevent,sec:sharpening}, we
consider $A$ run for at most $t$ queries and an \emph{isolated epoch run}
of the walkers, started with empty memory on the ordering $R$, for at
most $2t$ rounds. The competitor starts at most $t$ roots in $t$
queries, and the walker used to reproduce an edge from $r_j$ only needs
the prefix $r_1,\ldots,r_j$. Thus the good event need only concern
$r_1,\ldots,r_t$.

\subsection{Boxes and the good event}
\label{subsec:boxes}

Let $I := \lceil \log t \rceil$. For each scale $0 \le i \le I$ define
the \emph{box}
\[
B_i := \Bigl\{ k \in [t] : \type(r_k) \ge i \text{ and }
\#\{k' < k : \type(r_{k'}) \ge i\} < t/2^i \Bigr\}.
\]
In words, $B_i$ collects the first (roughly) $t/2^i$ stream positions
whose walks survive for at least $2^i$ steps; the budget argument in the
next subsection shows that these are the only positions that any
algorithm with budget $t$ can afford to explore to depth $2^i$. Three
immediate properties: $B_0 = [t]$ (every type is $\ge 0$ and the capacity
is $t$); $|B_i| \le \lceil t/2^i \rceil$; and each $B_i$ is a function of the type
vector $(\type(r_1), \dots, \type(r_t))$. For $j \in [t]$ let
\[
i_j := \max\{ i \le I : j \in B_i \},
\]
which is well defined since $j \in B_0$.

We use the following sufficient good event: no root lies in a region
protected by another root.
\[
G := \Bigl\{ \text{for all } i \le I,\ k \in B_i,\ 0 \le \ell \le 2^{i+1},\
j' \in [t] \setminus \{k\} : r_{j'} \neq r_k +_f \ell \Bigr\}.
\]
In words: for every scale $i$ and every box position $k \in B_i$, no root
other than $r_k$ itself lies on the first $2^{i+1}$ steps of the walk
from $r_k$. The condition with $\ell=0$ is immediate because $R$ is a
permutation. Note that $G$ is determined by $f$, $t$ and $R$ alone;
the algorithm $A$ does not enter its definition. We will show that $G$
holds with probability at least $4/5$ whenever
$t \le \sqrt{n}/(15 \log n)$ (\cref{lem:goodevent-warm}), and, by a
sharper argument, whenever $t \le \sqrt{n/(60\log n)}$
(\cref{lem:goodevent-sharp} in \cref{sec:sharpening}); and that on $G$
the epoch run dominates $A$ (\cref{cor:transfer}).

\subsection{The exploration lemma}
\label{subsec:exploration}

The next lemma describes the execution of $A$ on the good event: each
started root grows its own walk, walks never touch one another, and older
walks are longer. The last invariant below, a bound on the reach of each
walk, is where the budget enters: to push one walk to depth $2^{i+1}$,
the oldest-first rule forces every older open walk to at least that
length first, so either fewer than $t/2^i$ earlier stream positions
survive to depth $2^i$ --- placing the current root inside the box $B_i$
--- or $A$ has already exhausted its $t$ queries. The good event then
guarantees that the protected neighborhoods of box positions, which by
this budget argument contain everything $A$ explores, avoid all other
roots.

\begin{lemma}[Exploration lemma]
\label{lem:exploration}
Condition on $G$ and consider the execution of $A$ for up to $t$ queries.
At every point in time, the following invariants hold, where $\lambda_j$
denotes the number of edges discovered from the $j$-th root (the
\emph{reach} of $r_j$).
\begin{enumerate}
\item[(a)] The $j$-th fresh-label query, if it occurs, is made at $r_j$
(the scan of $R$ makes no skip). The explored graph is a vertex-disjoint
union of components, one per started root. The component of $r_j$ is
either an open path
$r_j \to r_j +_f 1 \to \dots \to r_j +_f \lambda_j$, or closed: its last
discovered edge entered an earlier vertex of the same path --- namely
$r_j$ itself, closing a pure cycle (so $\lambda_j = \rho(r_j)$), or an
interior vertex, in which case a collision was found and $A$ halted. In
particular, no discovered edge points into any root $r_{j'}$,
$j' \in [t]$, except the cycle-closing edge of $r_{j'}$'s own component.
\item[(b)] If $j < j'$ and both components are open paths, then
$\lambda_j \ge \lambda_{j'}$.
\item[(c)] $\lambda_j \le 2^{i_j + 1}$ for every started root $r_j$.
\end{enumerate}
\end{lemma}

\begin{proof}
We argue by induction on the queries of $A$. For a fresh-label query, it
is convenient to separate two conceptual substeps: assigning the unseen
label to a vertex of $f$, and then revealing the out-edge of that vertex.
Only the second substep is an oracle query.

Suppose first that $A$ is about to start its $j$-th root, where $j\le t$.
By the inductive hypothesis, every previously explored vertex has the
form $r_k+_f \ell$ with $k\in B_{i_k}$ and
$0\le\ell\le\lambda_k\le2^{i_k+1}$. The definition of $G$, applied
with $k$, $i=i_k$, and $j'=j$, shows that the next raw permutation
element $r_j$ is not among these vertices. The scan therefore does not
skip it, and the new label is assigned to $r_j$. We may momentarily view
it as a fresh path of length zero; (a)--(c) hold at this conceptual
substep. The query now reveals its first edge and is covered by the
extension analysis below with $\ell=1$.

Now suppose that $A$ extends the open path of $r_j$ from length
$\ell - 1$ to length $\ell$, discovering the vertex $v := r_j +_f  \ell$.

We first verify that (c) cannot be violated, i.e., that
$\ell \le 2^{i_j + 1}$. Suppose not, so $\ell = 2^{i_j + 1} + 1$
(invariants held before the step). If $i_j = I$ then
$\ell > 2^{I+1} \ge 2t$, exceeding $A$'s total budget of $t$ queries ---
impossible. So $i_j < I$. Since the path of $r_j$ was open at length
$2^{i_j+1}$, the walk from $r_j$ does not repeat within its first
$2^{i_j+1}$ steps, so $\rho(r_j) > 2^{i_j+1}$ and
$\type(r_j) \ge i_j + 1$. As $i_j$ is maximal, $j \notin B_{i_j+1}$,
which by the definition of the box forces its capacity to be exhausted:
\[
\#\{k' < j : \type(r_{k'}) \ge i_j + 1\} \;\ge\; t/2^{i_j+1}.
\]
Consider any such $k' < j$ at the present moment. If its component is
closed, its reach is $\lambda_{k'} = \rho(r_{k'}) \ge 2^{i_j+1}$ (as
$\type(r_{k'}) \ge i_j + 1$), so $A$ spent at least $2^{i_j+1}$ queries
on it. If its component is an open path then, by (b),
$\lambda_{k'} \ge \lambda_j = 2^{i_j+1}$, again at least $2^{i_j+1}$
queries. Summing over the at least $t/2^{i_j+1}$ such roots, and adding
the $2^{i_j+1} \ge 1$ queries spent on the path of $r_j$ itself, $A$ has
made more than $(t/2^{i_j+1}) \cdot 2^{i_j+1} = t$ queries --- a
contradiction. Hence $\ell \le 2^{i_j+1}$ and (c) is maintained.

Now consider the discovered vertex $v = r_j +_f  \ell$ with
$1 \le \ell \le 2^{i_j+1}$. Since $j \in B_{i_j}$, the event $G$
guarantees $v \neq r_{j'}$ for every $j' \in [t] \setminus \{j\}$: the
new edge does not enter any other root. Three cases remain. (i) $v$ is a
new vertex: the path of $r_j$ extends; (a) is clear. If $\ell=1$ is the
first edge of a fresh query, every older open path already has length at
least one, and (b) follows. Otherwise $A$ is oldest-first: it chose to
extend a path of length $\ell-1$ and picked the oldest such. Hence every
$j'<j$ with an open path had length at least $\ell-1$ before the query,
and none had length exactly $\ell-1$; thus its length is at least
$\ell=\lambda_j$. This proves (b). (ii) $v$ belongs to the component of
$r_j$: if $v = r_j$, the
component closes into a pure cycle with $\lambda_j = \ell = \rho(r_j)$;
no collision is created ($r_j$'s only known in-edge is the cycle edge)
and all invariants are maintained. If $v = r_j +_f  m$ for some
$1 \le m \le \ell - 1$, then $v$ now has two known in-edges, from
$r_j +_f  (m-1)$ and from $r_j +_f  (\ell-1)$ --- distinct vertices, as the
path $r_j, \dots, r_j +_f  (\ell - 1)$ is simple --- so $A$ has found a
collision and halts, as allowed by (a). (iii) $v$ belongs to the
component of another root $r_k$, $k \neq j$: since $v$ is not a root,
$v = r_k +_f  m$ with $m \ge 1$, and $v$ already has the known in-edge from
$r_k +_f  (m-1)$, a vertex different from $r_j +_f  (\ell - 1)$ by the
disjointness of components. Again $A$ has found a collision and halts.
\end{proof}

Invariants (a) and (c) immediately give:

\begin{corollary}[Coverage]
\label{cor:coverage}
On $G$, every vertex explored by $A$ within its first $t$ queries has the
form $r_k +_f  \ell$ with $k \in B_{i_k}$ and $0 \le \ell \le 2^{i_k + 1}$.
\end{corollary}

\subsection{The domination lemma}
\label{subsec:domination}

We now show that, on the good event, everything $A$ discovers is
discovered --- at most twice as slowly --- by the coupled epoch run of
the walkers. The walker responsible for an edge at depth $\ell$ of some
walk is the one whose cap $2^m$ is the first power of two at or above
$\ell$. The competitor and this walker process the same stream prefix,
and the oldest-first rule lets us compare their costs root by root: on
each earlier root, the competitor has paid at least
$\min(\ell, \rho)$ queries --- its open walks are at least as long as
$\ell$, and its closed walks paid their full $\rho$ --- while the walker
pays $\min(2^m, \rho) \le 2\min(\ell, \rho)$ units. The walker therefore
arrives at most a factor of two late.

\begin{lemma}[Domination]
\label{lem:domination}
Condition on $G$. Suppose that within its first $t' \le t$ queries, $A$
discovers the edge $\bigl(r_j +_f  (\ell-1)\bigr) \to (r_j +_f \ell)$, for
some $j \le t$ and $\ell \ge 1$, as the $\ell$-th edge of $r_j$'s
component. Then by the end of round $2t'$ of the coupled isolated epoch
run, this edge is in the memory $M$, unless a collision has already been
reported.
\end{lemma}

\begin{proof}
Let $m := \lceil \log \ell \rceil$, so that $\ell \le 2^m < 2\ell$ and
$m \le \lceil \log t \rceil \le L$: the walker $W_m$ exists. We track
$W_m$ in the isolated epoch run, assuming no collision report interrupts
it (otherwise we are done).

First, $W_m$'s walk from $r_j$, if reached, traverses the required edge:
by \cref{lem:exploration}(a), $A$'s component of $r_j$ has $\ell$
discovered edges, so the walk from $r_j$ does not self-repeat before step
$\ell$, i.e., $\rho(r_j) \ge \ell$; since also $2^m \ge \ell$, walker
$W_m$ walks at least $\ell$ steps from $r_j$ (its $\ell$-th step being
exactly the edge in question --- possibly a free lookup, but the edge is
then already in $M$). Recall from \cref{sec:prelim} that merging into
previously explored territory does not stop a walker; only its cap and
self-repeats of the current walk do. Hence the number of units $W_m$
spends on the walk from $r_{j'}$ is exactly
$\min(2^m,\rho(r_{j'}))$, one for each step.

Thus, by the end of round
\[
R := \sum_{j' < j} \min\bigl(2^m, \rho(r_{j'})\bigr) + \ell,
\]
walker $W_m$ has processed $r_1, \dots, r_{j-1}$ and walked $\ell$ steps
from $r_j$. It remains to show $R \le 2t'$.

Consider the moment at which $A$ queried the edge in question, extending
$r_j$'s path from length $\ell - 1$ to $\ell$; by then $A$ had made at
most $t'$ queries. These include the $\ell$ queries on $r_j$'s component
and, for each $j'<j$, at least
$\rho(r_{j'}) \ge \min(\ell, \rho(r_{j'}))$ of them if the component was
closed by that moment, and at least $\ell \ge \min(\ell, \rho(r_{j'}))$
of them if it was open (by the oldest-first property, as in case (i) of
\cref{lem:exploration}: at the extension moment every older open path had
length at least $\ell$). Hence
\[
t' \;\ge\; \sum_{j' < j} \min\bigl(\ell, \rho(r_{j'})\bigr) + \ell.
\]
Finally, $\min(2^m, x) \le 2\min(\ell, x)$ for every $x \ge 1$: if
$x \le \ell$ this reads $\min(2^m, x) = x \le 2x$, and if $x > \ell$ it
reads $\min(2^m, x) \le 2^m < 2\ell$. Therefore
\[
R = \sum_{j' < j} \min\bigl(2^m, \rho(r_{j'})\bigr) + \ell
\;\le\; 2 \sum_{j' < j} \min\bigl(\ell, \rho(r_{j'})\bigr) + \ell
\;\le\; 2t'. \qedhere
\]
\end{proof}

\begin{corollary}[Collision transfer]
\label{cor:transfer}
Condition on $G$, and suppose $A$ outputs a collision within $t$ queries.
Then the coupled isolated epoch run reports a collision within $2t$
rounds.
\end{corollary}

\begin{proof}
By \cref{obs:collision-witness} and \cref{lem:exploration}, $A$'s
collision is found at the moment its last query --- made at some time
$t' \le t$ --- discovers an already-explored, non-root vertex $v$ (cases
(ii)--(iii) of the extension step). The two witness edges are then
$e_1 = \bigl(r_k +_f  (m-1)\bigr) \to v$ with $v = r_k +_f  m$, $m \ge 1$,
discovered by $A$ at some time $\le t'$, and
$e_2 = \bigl(r_j +_f  (\ell-1)\bigr) \to v$ with $v = r_j +_f  \ell$,
$\ell \ge 1$, discovered at time $t'$; their sources are distinct
vertices. By \cref{lem:domination}, each of $e_1, e_2$ is in $M$ by the
end of round $2t'$ of the epoch run (or a collision was reported even
earlier, and we are done). Since the epoch run checks $M$ for collisions
at the end of every round, it reports one by the end of round
$2t' \le 2t$.
\end{proof}

\subsection{From constant probability to expectation}
\label{subsec:expectation}

The coupling gives a constant-probability guarantee for a single epoch.
Restarting from scratch would waste it --- the budget $t$ is not known to
the algorithm --- so $\Aall$ doubles the epoch length instead, refreshing
the stream but keeping the memory. Fresh streams make the epochs
independent trials; retained memory can only help, because walker
trajectories do not depend on the memory and the collision check is
monotone in it. The expected number of rounds is then a geometric series
over epochs.

\begin{proposition}[Epochs]
\label{prop:expectation}
Let $t \ge 1$ and let $\bar p$ be the probability that an isolated epoch
run of the walkers --- fresh stream, empty memory --- reports a collision
within $2t$ rounds. If $\bar p > 1/2$, then the expected number of rounds
of $\Aall$ on $f$ is at most $8t/(2\bar p - 1)$, and its expected number
of queries is at most $(\log n + 2) \cdot 8t/(2\bar p - 1)$.
\end{proposition}

\begin{proof}
Call epoch $e$ \emph{long} if $2^e \ge 2t$, i.e.,
$e \ge e^* := \lceil \log(2t) \rceil$, and \emph{successful} if the
isolated run on that epoch's stream (with empty memory) reports a
collision within $2t$ rounds. Since streams are fresh, the epochs'
successes are independent, each with probability $\bar p$.

We claim that within any successful long epoch, the actual algorithm ---
which enters the epoch with whatever memory $M_0$ it accumulated --- also
reports a collision within the epoch's first $2t \le 2^e$ rounds (unless
it has terminated even earlier). Indeed, walker trajectories and timings
are memory-independent (\cref{sec:prelim}), so the actual run and the
isolated run traverse identical edges in identical rounds; the actual
memory at the end of each round is a superset of the isolated one
($M_0 \cup M_{\mathrm{iso}} \supseteq M_{\mathrm{iso}}$); and the
reporting condition --- two edges of the memory into one vertex --- is
monotone under taking supersets. (The ``no collision'' exit cannot fire
first on an input with a collision: a memory containing all $n$ edges
contains a collision witness, which is checked first.)

Let $E$ be the first successful epoch among $e^*, e^*+1, \dots$; then
$\Aall$ terminates within $\sum_{e \le E} 2^e < 2^{E+1}$ rounds, and
$E = e^* + m$ with probability at most $(1 - \bar p)^m \bar p$. Hence,
using $2(1 - \bar p) < 1$,
\[
\Ex[\text{rounds}]
\le \sum_{m \ge 0} (1 - \bar p)^m \bar p \cdot 2^{e^* + m + 1}
= 2^{e^*+1} \bar p \sum_{m \ge 0} \bigl( 2(1 - \bar p) \bigr)^m
= \frac{2^{e^*+1} \bar p}{2\bar p - 1}
\le \frac{8t}{2\bar p - 1},
\]
since $2^{e^*} < 4t$. Each round costs at most $\log n + 2$ queries.
\end{proof}

\subsection{From the good event to instance optimality}
\label{subsec:master}

We can now assemble the pieces. The next proposition isolates the role of
the good event: any threshold up to which $G$ can be guaranteed with
probability $4/5$ yields instance optimality on the corresponding family
of inputs. It will be applied twice, with the first-moment bound of
\cref{subsec:goodevent} and with the sharper bound of
\cref{sec:sharpening}.

\begin{proposition}[Master proposition]
\label{prop:master}
Let $T = T(n)$ satisfy $1 \le T(n) \le n$, and suppose that for every
input $f \colon [n] \to [n]$ and every budget $1 \le t \le T(n)$, the
good event of \cref{subsec:boxes} satisfies $\Pr[G] \ge 4/5$. Then for
every $f \in \mathcal{F}_{T(n)/20}$ and every Las Vegas algorithm $A'$
for collision detection,
\[
\Queries_{\Aall}(f) \;\le\; 400\, (\log n + 2) \cdot
\max_{\pi} \Queries_{A'}(f^{\pi}).
\]
\end{proposition}

\begin{proof}
Fix $f \in \mathcal{F}_{T(n)/20}$ and a competitor $A'$; then
$\OPT(f) \le T(n)/20 < n$, so $f$ contains a collision
(\cref{obs:no-collision-witness}). Construct $A$, $q^*$ and $t = t(A)$ as
in \cref{subsec:setup}; by \eqref{eq:tbound},
$t \le 20\,\OPT(f) \le T(n)$, so $\Pr[G] \ge 4/5$ by hypothesis.

Let $\bar p$ be as in \cref{prop:expectation}. The event that $A$ outputs
a collision within $t$ queries has probability at least $9/10$ (by the
definition of $t$), over the joint distribution of the stream and $A$'s
remaining coins; intersecting it with $G$ (a stream-measurable event) and
applying \cref{cor:transfer} pointwise, every outcome in the intersection
has a stream on which the isolated epoch run reports a collision within
$2t$ rounds. Hence
\[
\bar p \;\ge\; \Pr\bigl[ A \text{ succeeds within } t \text{ queries} \bigr]
- \Pr[\neg G] \;\ge\; \frac{9}{10} - \frac{1}{5} \;=\; \frac{7}{10}.
\]
By \cref{prop:expectation},
\[
\Queries_{\Aall}(f)
\;\le\; (\log n + 2) \cdot \frac{8t}{2 \cdot \frac{7}{10} - 1}
\;=\; 20\, t (\log n + 2).
\]
Combining with
$t \le 10 q^* \le 20\,\OPT(f) \le 20 \max_\pi \Queries_{A'}(f^\pi)$ from
\eqref{eq:tbound} completes the proof.
\end{proof}

\subsection{A first bound on the good event}
\label{subsec:goodevent}

For $0 \le i \le \lfloor \log n \rfloor$ let
$\Type_i := \{v \in [n] : \type(v) = i\}$ and let
\[
N_i := \#\{ j \in [t] : \type(r_j) = i \}
\]
count the stream positions of type $i$.

\begin{definition}[Overpopulation]
\label{def:overpopulation}
Type $i$ is \emph{overpopulated} if
$N_i > 10 \log(2n) \cdot t |\Type_i| / n$.
\end{definition}

Since $\Ex[N_i] = t |\Type_i| / n$, Markov's inequality gives
$\Pr[\text{type $i$ overpopulated}] \le \frac{1}{10 \log(2n)}$, and as
the number of types is at most
$\lfloor \log n \rfloor + 1 \le \log(2n)$, a union bound gives
\begin{equation}
\label{eq:overpopulation}
\Pr[\text{some type is overpopulated}] \;\le\; \frac{1}{10}.
\end{equation}

\begin{lemma}[Good event, first-moment bound]
\label{lem:goodevent-warm}
If $t \le \sqrt{n}/(15 \log n)$ and $n \ge 2^{24}$, then
$\Pr[G] \ge 4/5$.
\end{lemma}

\begin{proof}
Condition on the type vector
$\tau = (\type(r_1), \dots, \type(r_t))$ and assume no type is
overpopulated under $\tau$; by \eqref{eq:overpopulation} this
conditioning costs probability at most $1/10$. Conditioned on $\tau$,
the positions of each type form a uniform ordered sample without
replacement from that type class, and every box $B_i=B_i(\tau)$ is
determined by $\tau$.

Enumerate the coincidence events: for $0 \le i \le I$,
$k \in B_i(\tau)$ and $0 \le \ell \le 2^{i+1}$, let $E_{i,k,\ell}$ be the
event that $r_{j'} = r_k +_f  \ell$ for some $j' \in [t] \setminus \{k\}$;
then $\neg G = \bigcup E_{i,k,\ell}$. Fix one event and further condition
on $r_k$, so that $w:=r_k+_f \ell$ is fixed, and put $h:=\type(w)$. If
$w=r_k$, then $E_{i,k,\ell}$ is impossible because the roots are
distinct. Otherwise, if $\type(r_k)\ne h$, the probability that $w$
appears in one of the $N_h$ positions of type $h$ is $N_h/|\Type_h|$.
If $\type(r_k)=h$ and $|\Type_h|>1$, this probability is
$(N_h-1)/(|\Type_h|-1)\le N_h/|\Type_h|$; when
$|\Type_h|=1$, the case $w\ne r_k$ is impossible. Hence, by
non-overpopulation,
\[
\Pr[ E_{i,k,\ell} \mid \tau, r_k ]
\le \frac{N_h}{|\Type_h|}
\le \frac{10 \log(2n)\, t}{n}.
\]
The number of events is at most (using $|B_i| \le t/2^i + 1$,
$2^I \le 2t$, and $I + 1 \le \log t + 2 \le 2t$)
\[
\sum_{i=0}^{I} |B_i| \bigl( 2^{i+1} + 1 \bigr)
\le \sum_{i=0}^{I} \Bigl( 2t + \frac{t}{2^i} + 2^{i+1} + 1 \Bigr)
\le 2t(I+1) + 2t + 8t + (I+1)
\le 2t \bigl( \log t + 8 \bigr).
\]
Combining, for every non-overpopulated $\tau$,
\[
\Pr[\neg G \mid \tau]
\le 2t (\log t + 8) \cdot \frac{10 \log(2n)\, t}{n}
= \frac{20\, t^2 \log(2n) (\log t + 8)}{n}.
\]
For $t \le \sqrt{n}/(15 \log n)$ and $n \ge 2^{24}$ we have
$\log t \le \frac{1}{2} \log n$, hence
$\log t + 8 \le (\frac{1}{2} + \frac{1}{3}) \log n = \frac{5}{6} \log n$
(as $8 \le \frac{1}{3} \log n$), and
$\log(2n) \le \frac{25}{24} \log n$; so
\[
\Pr[\neg G \mid \tau]
\le \frac{20 \cdot \frac{25}{24} \cdot \frac{5}{6} \log^2 n}{225 \log^2 n}
\le \frac{1}{10}.
\]
Adding the overpopulation probability \eqref{eq:overpopulation},
$\Pr[\neg G] \le 1/10 + 1/10 = 1/5$.
\end{proof}

\subsection{Proof of the warm-up theorem}
\label{subsec:together}

\begin{proof}[Proof of \cref{thm:warmup}]
Set $c_0 := 1/300$ and $T_0(n) := \sqrt{n}/(15 \log n)$, so that
$\Feasywarm = \mathcal{F}_{T_0(n)/20}$. For $n < 2^{24}$ the claim holds
by adjusting the hidden constant in the $O(\log n)$ factor: by
\cref{rem:termination}, $\Queries_{\Aall}(f)$ is bounded by a constant
depending only on $n$, while
$\max_\pi \Queries_{A'}(f^\pi) \ge \OPT(f) \ge 1$ for every competitor
$A'$ (an algorithm making no queries cannot be correct on all inputs).
For $n \ge 2^{24}$, \cref{lem:goodevent-warm} verifies the hypothesis of
\cref{prop:master} for the threshold $T_0$, and the proposition gives
$\Queries_{\Aall}(f) \le 400 (\log n + 2) \cdot
\max_\pi \Queries_{A'}(f^\pi)$ for every $f \in \Feasywarm$ and every
Las Vegas competitor $A'$, which is the claimed $O(\log n)$-instance
optimality. (We made no attempt to optimize the constants.)
\end{proof}

\section{Sharpening the threshold}
\label{sec:sharpening}

Everything in \cref{sec:proof} except \cref{lem:goodevent-warm} is
insensitive to the exact bound on $t$: the coupling, the exploration and
domination lemmas, and \cref{prop:master} use only the definition of the
good event, not its probability. To raise the threshold of
\cref{thm:warmup} we prove a stronger bound on
$\Pr[G]$.

We start with the proof intuition.  For a scale $i$, let
\[
 P_i(u):=\{u,u+_f 1,\ldots,u+_f 2^{i+1}\}
\]
be the prefix protected at $u$.  Once $u$ is fixed, a new uniform root
hits $P_i(u)$ with probability at most $|P_i(u)|/(n-t+1)$.  Thus exposing
the roots one by one easily controls the event that a new root enters a
previously exposed walk.

The opposite direction has no such bound.  If a vertex $v$ was exposed
earlier and the new root is $u$, then the relevant quantity is the number
of starting vertices whose walk reaches $v$.  This reverse neighborhood
can be arbitrarily large: if $f(x)=v$ for every $x$, every starting
vertex reaches $v$ in one step.
To circumvent this asymmetry, for each
possible target $r_{j}$, we first reveal all other randomness in our structure, and then reveal $r_{j}$ last and ask what is the probability that it breaks the good event. The key idea is that the size of the danger set can be suitably bounded, whereas the location of $r_j$ is uniformly random, which bounds the probability of $r_j$ to fall within the danger set. 
The proof below formalizes this argument.

\begin{lemma}[Good event, sharpened]
\label{lem:goodevent-sharp}
If $t \le \sqrt{n/(60 \log n)}$ and $n \ge 2^{12}$, then
$\Pr[G] \ge 9/10$.
\end{lemma}

\begin{proof}
{\revisioncolor
Fix a possible target index $j$ and expose all roots except $r_{j}$.
At scale $i$, call a position $i$-long if its walk survives for at least
$2^i$ steps.  Scan the positions in their original order, skip $j$, and
take the first $\lceil t/2^i\rceil$ $i$-long positions (or all of them if
there are fewer).  Denote this surrogate box by
$\widetilde B_i^{(-j)}$.  It is determined by the exposed roots.

The original $B_i$ consists of the first $\lceil t/2^i\rceil$ $i$-long
positions without skipping $j$.  Deleting a position can only decrease
the ranks of the remaining positions, and therefore
\[
 B_i\setminus\{j\}\subseteq\widetilde B_i^{(-j)}.
\]

Take the same protected prefixes as in the definition of $G$ and let
\[
\widetilde{S}^{(-j)} := \Bigl\{ r_k +_f  \ell : 0 \le i \le I,\
k \in \widetilde{B}^{(-j)}_i,\ 0 \le \ell \le 2^{i+1} \Bigr\},
\]
which is fixed before $r_{j}$ is exposed.  Its size is at most
\[
 |\widetilde S^{(-j)}|
 \le \sum_{i=0}^{I}
       \Bigl\lceil\frac{t}{2^i}\Bigr\rceil(2^{i+1}+1)
 \le 2t(\log t+8).
\]
Indeed, after using $\lceil t/2^i\rceil\le t/2^i+1$, the leading product
contributes $2t$ at each scale, and the remaining geometric and rounding
terms contribute at most $12t$ in total.

If $G$ fails with $r_{j}$ as the target, then for some $i$, some
$k\in B_i\setminus\{j\}$, and some $0\le\ell\le2^{i+1}$, we have
$r_{j}=r_k+_f \ell$.  The inclusion above puts $k$ in the surrogate box,
so $r_{j}\in\widetilde S^{(-j)}$.

Now expose $r_{j}$.  Conditioned on the other roots, it is uniform among
the $n-t+1$ unused vertices.  Since the surrogate set is already fixed,
\[
 \Pr\bigl[r_{j}\in\widetilde S^{(-j)}
           \mid \text{the other roots}\bigr]
 \le \frac{2t(\log t+8)}{n-t+1}.
\]
Every failure has some root as its target.  A union bound over the $t$
possible target indices thus gives
\[
\Pr[\neg G]
\le\frac{2t^2(\log t+8)}{n-t+1}.
\]

Finally, the assumptions give $t\le n/2$ and
$t\le\sqrt{n/(60\log n)}$. Moreover,
$\log t\le\frac12\log n$, and $n\ge2^{12}$ gives $8\le\log n$.
Consequently $\log t+8\le\frac32\log n$, and
\[
\Pr[\neg G] \le \frac{4t^2(\log t+8)}{n}
\le \frac{6t^2\log n}{n}
\le \frac{6\log n}{60\log n}=\frac{1}{10}. \qedhere
\]
}
\end{proof}

The main theorem follows by the same assembly as before, with the
improved lemma in place of \cref{lem:goodevent-warm}.

\begin{proof}[Proof of \cref{thm:main}]
Set $c := 1/160$ and $T_1(n) := \sqrt{n/(60 \log n)}$. Since
\[
20c\sqrt{n/\log n}
= \sqrt{n/(64\log n)}
\le T_1(n),
\]
we have
$\Feasy \subseteq \mathcal{F}_{T_1(n)/20}$. For $n < 2^{12}$ the claim
holds by adjusting the hidden constant, exactly as in the proof of
\cref{thm:warmup}. For $n \ge 2^{12}$, \cref{lem:goodevent-sharp}
verifies the hypothesis of \cref{prop:master} for the threshold $T_1$
(with room to spare: $\Pr[G] \ge 9/10 \ge 4/5$), and the proposition
gives
$\Queries_{\Aall}(f) \le 400 (\log n + 2) \cdot
\max_\pi \Queries_{A'}(f^\pi)$ for every
$f \in \mathcal{F}_{T_1(n)/20} \supseteq \Feasy$ and every Las Vegas
competitor $A'$.
\end{proof}

\section{Open Problems}
\label{sec:questions}

\paragraph{The collision conjecture and the merging barrier.}
\Cref{thm:main} proves \cref{conjecture:BGN24} for $\Feasy$; proving the $O(\log n)$-instance optimality in the rest of the range, including above the birthday threshold, remains open. We note that this situation is similar to that of claw detection in graphs \cite{BGN24}. For both collisions and claws, the conjecture is known to be true only in the setting where merging between different walks provably cannot happen. (Recall that two walks merge if the head of one of the walks reaches the tail of the other, making them effectively concatenate into a single, longer walk.)

Unfortunately, current proof approaches in unlabeled instance optimality, including the use of the ``oldest-first principle'' in our proof, crucially rely on having no merges at all throughout the algorithm's run. Merges are common and sometimes unavoidable above the birthday threshold: making $q = \Omega(\sqrt{n})$ random fresh queries in a set of size $n$ will produce $\Omega(q^2 / n)$ merges, which breaks down current techniques. Thus, proving the conjecture in full generality requires techniques that are able to handle merging walks. We leave this as an intriguing open problem.

\paragraph{Single-scale optimality.}
\Cref{thm:main} leaves open a stronger and arguably more interesting
possibility: that for every input there is a \emph{single} scale $i$ such
that the lone walker $W_i$ is $O(1)$-competitive with the structure-aware
optimum on that input, with $\Aall$ then paying only the $O(\log n)$
overhead of running all scales side by side. That is, we ask the following.

\begin{center}
\begin{quote}
\emph{
Is there an absolute constant $C$ such that for every $f$ with a collision
there is a scale $i$ for which $W_i$ alone, run on its own stream, finds
a collision within $C \cdot \OPT(f)$ queries with probability at least
$2/3$?
}
\end{quote}
\end{center}

If true, the logarithmic
factor in \cref{thm:main} would be exactly the price of
\emph{universality} --- the cost of not knowing which scale fits the
instance --- rather than a loss inherent in walk-based algorithms.
As the $\Omega(\log n)$ lower bound for instance optimality is known to apply in our regime $\Feasy$ \cite{BGN24}, a positive answer to the
single-scale question would identify this overhead precisely with the
cost of running all possible scales.

Our proof does not give this. The domination lemma
(\cref{lem:domination}) matches each edge that the competitor discovers
to the walker at the matching scale, and different edges --- including
the two edges of a single collision witness --- may be matched to
different walkers. One might hope to use only the largest relevant scale,
since a walker with a larger cap traverses everything a smaller one does;
but it also spends more time per starting vertex: on stream positions
whose walks survive long (large $\rho$-value), a large-scale walker pays
up to its full cap before moving on, and the factor-two accounting of
\cref{lem:domination} breaks. A candidate hard instance would combine
components at multiple well-separated scales, so that the witness needs both,
with many long-surviving vertices that slow the larger scale down.
Whether such an instance defeats every single scale by an $\omega(1)$
factor, or whether some averaging argument always identifies one good
scale, is left open.

\paragraph{Beyond collisions and beyond unlabeled instance optimality.}
Which other search problems admit similar guarantees? \cite{BGN24}
show that $3$-way collision detection and fixed-point detection are
polynomially far from instance optimal, under the unlabeled definition of instance optimality used here (and in \cite{GrossmanKN20,BGN24,HsiangL23}). For problems such as $3$-collision and fixed-point detecion, it would be interesting to
identify structural conditions or alternative, weaker notions of instance optimality under which a universal algorithm is
near-optimal below the corresponding random-instance threshold.

\section*{AI Disclosure}
The mathematical content, including all proofs, was first conceived and written by the human authors, without the use of LLMs; but see comment about Section 5 below. After that we used LLMs, including Anthropic Claude Fable 5 and OpenAI ChatGPT Sol 5.6, for extensive editorial purposes throughout the text. They were used for text editing, polishing and clarifying of entire sections throughout the paper, but were instructed to keep the human-generated mathematical core and ideas. The LLMs were also tasked with drafting the first version of the introduction (which was then heavily edited by the human authors), and for literature review. 

Along the way, an LLM (Claude Fable 5) found an inaccuracy in our use of randomness in the proof of Section $5$ (improvement from $\sqrt{n} / \log n$ to $\sqrt{n / \log n}$) and suggested a correction. Thus, the credit for all mathematical ideas is due to the human authors, except for Section $5$ where both the human authors and the LLM contributed meaningfully. The human authors take full responsibility for all parts of the manuscript.

\bibliographystyle{plain}
\bibliography{references}

\appendix

\section{Deferred proofs from \texorpdfstring{\cref{sec:standard}}{Section 3}}
\label{app:reduction}

\begin{proof}[Proof of \cref{lem:relabel}]
$A''$ samples a uniformly random permutation $\pi \in S_n$ lazily and
simulates $A'$ on the input
$f^\pi=\pi^{-1}\circ f\circ\pi$. When $A'$ queries a label $x$,
$A''$ queries $f$ at $\pi(x)$, obtains $v=f(\pi(x))$, and returns
$\pi^{-1}(v)$ to $A'$. The values of $\pi$ and $\pi^{-1}$ are assigned
lazily whenever a new label appears. This is a faithful simulation of
$A'$ on $f^\pi$, using exactly one oracle query per simulated query.
Correctness transfers: $x \neq y$ is a collision of $f^\pi$ if and only
if $\pi(x) \neq \pi(y)$ is a collision of $f$, and $A''$ translates the
witness (or the answer ``no collision'') accordingly. Hence $A''$ is Las
Vegas and
$\Queries_{A''}(f) = \Ex_\pi \Queries_{A'}(f^\pi)
\le \max_\pi \Queries_{A'}(f^\pi)$.

Finally, $A''$ is label-oblivious: for any fixed unlabeled transcript,
the labels that $A''$ observes are, by the uniformity of $\pi$, a
uniformly random injective assignment of concrete names to the
placeholders, independent of the history of $A'$'s decisions; formally,
the distribution of $A''$'s next simulated operation given its unlabeled
transcript $H$ is the same for every labeling of $H$, so $A''$ can be
implemented as a randomized function of $H$ alone.
\end{proof}

\begin{proof}[Proof of \cref{lem:uniform}]
We induct on the number of queries. The claim is trivial for the empty
transcript. Condition on the current unlabeled transcript $H$ and on the
coins of $A$. The next abstract query is then fixed.

First suppose that the queried label has not appeared before. Introduce
a new placeholder $x$. Under the lazy random relabeling, $e(x)$ is
uniform among $[n]\setminus e(H)$. Equivalently, before revealing the
answer, the current embedding is uniform over all injective extensions
of the old embedding to $x$: every pair consisting of an old consistent
embedding and an unused image for $x$ has the same probability.

It remains to reveal the answer to the query, whether $x$ is new or was
already present. The answer is $f(e(x))$, a deterministic function of
the extended embedding. For each old placeholder $y$, the outcome
``the answer is $y$'' restricts the embeddings to
\[
E_H^y:=\{e:f(e(x))=e(y)\}.
\]
Conditioning a uniform distribution on this set leaves it uniform. For
the outcome ``new vertex'', introduce a new placeholder $z$. The map
\[
e\longmapsto e\cup\{z\mapsto f(e(x))\}
\]
is a bijection between the embeddings producing a new answer and the
consistent embeddings of the new transcript $H'$. Hence the new
embedding is uniform on $E_{H'}$ in every case.
\end{proof}

\begin{proof}[Proof of \cref{lem:standard}]
$A^\circ$ keeps an internal \emph{virtual transcript} distributed as the
transcript of $A$, together with a correspondence from virtual
components to the components it has actually explored. The
correspondence preserves the known directed graph and, for every
unmerged open path, its length. The virtual transcript includes the full
chronology, so all history-dependent decisions of $A$ can be reproduced
exactly.

Suppose first that the virtual execution starts a new root. Then
$A^\circ$ also queries an unseen label, and the outcome is transported
to the virtual transcript through the correspondence. Suppose next that
the virtual execution queries the endpoint of a path $P$. If the current
transcript is already merged, $A^\circ$ simply queries the corresponding
actual endpoint. If it is unmerged, let $\ell$ be the length of $P$ and
let $P^\circ$ be the oldest actual open path of length $\ell$. Before the
query, update the internal correspondence by mapping $P$ to $P^\circ$
and permuting the other equal-length paths accordingly. The actual query
then extends $P^\circ$, as required.

It remains to justify that this remapping does not change the law of the
answer. The explored edge constraints on two unmerged paths of the same
length are identical. Swapping their images is therefore a bijection of
the consistent embeddings. By \cref{lem:uniform}, these embeddings are
uniform even after conditioning on the full chronological transcript.
Consequently, querying the remapped path has the same distribution over
a new answer or a coincidence with any existing placeholder. We
transport the outcome back to the virtual transcript and update the
correspondence. This maintains the invariant by induction, including
when the query closes a path, creates a witness-free merge, or finds a
collision.

Thus the virtual transcript has exactly the law of $A$'s transcript,
while the actual execution is oldest-first whenever it is unmerged. The
two algorithms make the same number of queries and return the same answer
under the coupling, proving the lemma.
\end{proof}

\end{document}